\documentclass[10pt]{sigplanconf-pldi16}
\usepackage{SIunits}            % typset units correctly
\usepackage{courier}            % standard fixed width font
\usepackage[scaled]{helvet} % see www.ctan.org/get/macros/latex/required/psnfss/psnfss2e.pdf
\usepackage{url}                  % format URLs
\usepackage{listings}          % format code
\usepackage{enumitem}      % adjust spacing in enums
\usepackage[colorlinks=true,allcolors=blue,breaklinks,draft=false]{hyperref}   % hyperlinks, including DOIs and URLs in bibliography
\newcommand{\doi}[1]{doi:~\href{http://dx.doi.org/#1}{\Hurl{#1}}}   % print a hyperlinked DOI

\usepackage{amsmath,amsthm}
\usepackage{graphicx}
\usepackage{tikz}
\usepackage{color}
\usepackage{stmaryrd}
\usepackage{eufrak}

\theoremstyle{theorem}
\newtheorem{theorem}{Theorem}[section]
\theoremstyle{definition}
\newtheorem{example}{Example}[section]
\newtheorem{definition}{Definition}[section]

\begin{document}

\title{Local Lexing\thanks{This work has been funded by EPSRC grant 
\href{http://gow.epsrc.ac.uk/NGBOViewGrant.aspx?GrantRef=EP/L011794/1}{EP/L011794/1}.}}

%
% any author declaration will be ignored  when using 'pldi' option (for double blind review)
%

%\authorinfo{.}
\authorinfo{Steven Obua}
{University of Edinburgh}
{steven.obua@gmail.com}
\authorinfo{Phil Scott}
{University of Edinburgh}
{phil.scott@ed.ac.uk}
\authorinfo{Jacques Fleuriot}
{University of Edinburgh}
{jdf@inf.ed.ac.uk}

\maketitle

\begin{abstract}
We introduce a novel parsing concept called \emph{local lexing}. It integrates the classically separated stages of lexing and parsing by allowing lexing to be dependent upon the parsing progress and by providing a simple mechanism for constraining lexical ambiguity. This makes it possible for language design to be composable not only at the level of context-free grammars, but also at the lexical level. It also makes it possible to include lightweight error-handling directly as part of the language specification instead of leaving it up to the implementation.  

We present a high-level algorithm for local lexing, which is an extension of Earley's algorithm. We have formally verified the correctness of our algorithm with respect to its local lexing semantics in Isabelle/HOL.
\end{abstract}

\section{Introduction}
\label{sec:introduction}

\newcommand{\nonterminals}{\ensuremath{\mathfrak{N}}}
\newcommand{\terminals}{\ensuremath{\mathfrak{T}}}
\newcommand{\startsymbol}{\ensuremath{\mathfrak{S}}}
\newcommand{\rules}{\ensuremath{\mathfrak{R}}}
\newcommand{\lexer}{\ensuremath{\textsl{Lex}}}
\newcommand{\selector}{\ensuremath{\textsl{Sel}}}
\newcommand{\lang}{\ensuremath{\mathcal{L}}}
\newcommand{\prefixlang}{\ensuremath{\lang_\text{prefix}}}
\newcommand{\derives}{\ensuremath{\overset{*\ }{\Rightarrow}}}
\newcommand{\lderives}{\ensuremath{\ {\overset{*\ }{\Rightarrow}}_L}\ }
\newcommand{\terminalsof}[1]{\ensuremath{[#1]}}
\newcommand{\charsof}[1]{\ensuremath{\overline{#1}}}
\newcommand{\locallexing}{\ensuremath{\ell\ell}}
\newcommand{\lextoken}[2]{\small\ensuremath{\dfrac{\texttt{#1}}{\textsl{#2}}}}
\newcommand{\sele}{\ensuremath{<_\selector}}
\newcommand{\nterminal}[1]{\ensuremath{\textsl{#1}}}
\newcommand{\terminal}[1]{\ensuremath{\textsl{#1}}}
\newcommand{\closure}[3]{\ensuremath{\textsl{closure}_{#1}\, #2\ #3}}
\newcommand{\emptyseq}{\ensuremath{\varepsilon}}
\newcommand{\seqlen}[1]{\ensuremath{\left|{#1}\right|}}
\newcommand{\charslen}[1]{\seqlen{\charsof{#1}}}
\newcommand{\paths}[2]{\ensuremath{\mathcal{P}_{#1}^{#2}}}
\newcommand{\tokensat}[1]{\ensuremath{\mathcal{X}_{#1}}}
\newcommand{\admissibletokens}[2]{\ensuremath{\mathcal{W}_{#1}^{#2}}}
\newcommand{\selectedtokens}[2]{\ensuremath{\mathcal{Z}_{#1}^{#2}}}
\newcommand{\appendtokens}[1]{\ensuremath{\operatorname{Append}_{#1}}}
\newcommand{\limit}[2]{\ensuremath{\operatorname{limit}\ {#1}\ {#2}}}
\newcommand{\allpaths}{\ensuremath{\mathfrak{P}}}
\newcommand{\allitems}{\ensuremath{\mathfrak{I}}}

\noindent The traditional approach to specifying the syntax of a computer language is to define two components, a \emph{lexer} (also called \emph{scanner}) and a \emph{parser}. The lexer partitions the input document consisting of a sequence of \emph{characters} into a sequence of \emph{tokens}. Each token is uniquely associated with a \emph{terminal}, such that the sequence of tokens can be viewed as a sequence of terminals. The parser is typically defined by a \emph{context-free grammar} (CFG), and checks if the sequence of terminals is in the language generated by the CFG.  

CFGs are a powerful language design tool. A most important property of CFGs is composability. Given two or more CFGs, it is easy to combine them into a single CFG in various ways in order to specify a language as a mashup of several other languages, which is a common scenario in modern programming environments.  

The lexer component of this traditional setup is problematic though in such a mashup scenario. A keyword in one language might be an identifier in another language, rendering the lexers of these two languages incompatible with each other. The problem is that lexers are not supposed to be composed in the traditional setup. Practical solutions to this problem usually involve some form of ad-hoc communication between parser and lexer stages, thus shifting an issue which should be dealt with at the language design level to the implementation level. 

The main reasons for the traditional split of syntax recognition into lexing and parsing are \emph{speed}, \emph{expressivity} and \emph{convenience}:
\begin{description}
\item[Speed] Typically terminals are specified via regular expressions that can be recognized with deterministic finite state machines. This is usually much faster than parsing with respect to a CFG. 
\item[Expressivity] Despite CFGs being more expressive than regular expressions, two common lexing rules cannot be expressed with CFGs: 
\emph{longest-match}, and \emph{priority}. The longest-match rule states that if multiple terminals are associated with regular expressions that could all match prefixes of the rest of the sequence of characters to be scanned, then the terminals which match the longest prefix are to be preferred. The priority rule comes into play if after application of the longest-match rule there are still at least two different terminals left as possible candidates: then a linear priority order among terminals is assumed, and the terminal with the highest priority among all candidates is picked. 
\item[Convenience] In this setup, whitespace and comments are usually special terminals which do not appear in any grammar rule, but which are filtered out of the sequence of terminals during the lexing stage.
\end{description} 

\noindent\emph{Scannerless parsing}~\cite{scannerless} proposes to solve our lexing problem by relinquishing separate lexing, effectively identifying characters, tokens and terminals. This negatively affects all three mentioned advantages of a separate lexing stage, the most severely affected being expressivity: in order to approximate the missing longest-match and priority rules, scannerless parsing introduces follow restrictions and reject productions, which are awkard additions to the elegant formalism of context-free grammars because they destroy the nice composability properties of context-free grammars. A scannerless parsing technique called \emph{packrat parsing} goes even further and does away with context-free grammars entirely, employing \emph{parsing expression grammars} instead~\cite{peg}, again to the detriment of composability. 

Instead, we propose a new parsing semantics which we call \emph{local lexing}. Local lexing keeps the distinction between lexing and parsing, and between characters, tokens and terminals. But instead of deterministically converting a sequence of characters into a sequence of tokens, local lexing converts a sequence of characters into a \emph{set} of token sequences, applying lexing and parsing in an intertwined manner. 

The contributions of this paper are as follows: We first define the novel concept of local lexing in Section~\ref{sec:lldef}. We then present examples of applications of local lexing in Section~\ref{sec:applications}. These examples show that
local lexing is a generalisation of the traditional setup, and that local lexing readily allows integrated access at the language design level to issues such as lexical composability and error-handling. In 
Section~\ref{sec:algorithm} we describe a high-level algorithm which implements local lexing as an extension of Earley's algorithm. We have formally verified the correctness of this algorithm in Isabelle/HOL~\cite{isabelle}, and provide an outline of this correctness proof in Section~\ref{sec:proof}. The full proof (a total of 11411 lines or about 230 pages) is available in \cite{locallexingtheories}. We also provide a practical library for local lexing~\cite{locallexingprototype}, written in \mbox{Scala/Scala.js}~\cite{scala,scalajs}. The library contains examples of its application to the examples in Section~\ref{sec:applications}. Before concluding, we discuss further related work in Section~\ref{sec:relatedwork}.

\section{Definition of Local Lexing}
\label{sec:lldef}

Before defining local lexing, we remind the reader of a few basic notions. For a set $U$ we let $U^*$ denote the set of sequences with elements in $U$,
$\emptyseq \in U^*$ denotes the empty sequence, and for two sequences $\alpha \in U^*$ and $\beta \in U^*$ we let $\alpha \beta \in U^*$ denote their
concatenation. Given a sequence $\alpha \in U^*$, 
we denote its length by $\seqlen{\alpha}$ and let $\alpha_i \in U$ denote the $i$-th element of $\alpha$ for $i \in \{0, \ldots, |\alpha| - 1\}$.
A context-free grammar is a quadruple $(\nonterminals, \terminals, \rules, \startsymbol)$, where $\nonterminals$ is the set of 
nonterminals, $\terminals$ the set of terminals (such that $\nonterminals$ and $\terminals$ are disjoint), 
$\rules \subseteq \nonterminals \times (\nonterminals \cup \terminals)^*$ the rules of the grammar and $\startsymbol \in \nonterminals$ the start symbol.
Instead of $(\nterminal{N}, \alpha) \in \rules$ we often write $\nterminal{N} \rightarrow \alpha$. 
For $\alpha, \beta \in (\nonterminals \cup \terminals)^*$ we say $\alpha \Rightarrow \beta$ iff there are $\alpha_0, \nterminal{N}, \alpha_1$ and $\gamma$
such that $\alpha = \alpha_0 \nterminal{N} \alpha_1$, $\beta = \alpha_0 \gamma \alpha_1$ and $\nterminal{N} \rightarrow \gamma$. Furthermore, we write \derives\ for the reflexive and transitive closure of $\Rightarrow$. We define the language $\lang$ of a grammar $G$ as the set of all sequences of terminals derivable from the 
start symbol, $\lang = \{ w \in \terminals^* \, | \, \startsymbol \derives w \}$. 
Furthermore we define the set of \emph{prefixes} 
$\prefixlang$ of $G$ via 
\[\prefixlang = \left\{ w \in \terminals^* \, | \,\exists\, \alpha \in (\nonterminals \cup \terminals)^*.\, \startsymbol \derives w\alpha \right\}.\]

\begin{definition}[Token]
Given a set of terminals $\terminals$, and a set of characters $\Sigma$, a \emph{token} $x$ is a pair $x = (t, c) \in \terminals \times \Sigma^*$.
In examples we will often use the notation \[\lextoken{$c$}{$t$}\] for the token $x$. 
We call the token \emph{empty} iff $|c| = 0$.
We define $\terminalsof{x} = t$ and $\charsof{x} = c$.
We lift these notations in the obvious manner to sequences of tokens: 
given a token sequence $q = x_0 \ldots x_r \in (\terminals \times \Sigma^*)^*$ we define
$\terminalsof{q} = \terminalsof{x_0} \ldots \terminalsof{x_r} \in \terminals^*$ and $\charsof{q} = \charsof{x_0} \ldots \charsof{x_r} \in \Sigma^*$.
\end{definition}

\begin{definition}[Local Lexing] 
Given a set of terminals $\terminals$, and a set of characters $\Sigma$, we call a pair $(\lexer, \selector)$ a \emph{local lexing} (with respect to 
$\terminals$ and $\Sigma$) iff:
\begin{itemize}
\item The \emph{lexer} $\lexer$ assigns to each terminal $t \in \terminals$ a lexing function $\lexer(t)$ which, given a sequence of characters $D \in \Sigma^*$ and a position $k \in \{0, \ldots, |D|\}$, returns a set consisting of tokens $(t, c)$ such that $k + |c| \leq |D|$ and
$c_i = D_{k + i}$ for all $i$ such that $0 \le i \le |c| - 1$.
\item The \emph{selector} $\selector$ takes two token sets $A$ and $B$ such that $A \subseteq B$ and returns a token set 
$\selector(A, B)$ such that $A \subseteq \selector(A, B) \subseteq B$. We usually define $\selector$ indirectly by defining a strict partial order $\sele$ on tokens and setting
\[ \selector(A, B) = A \cup \left\{ x \in B \mid \forall\, y \in B.\ \neg\ x \sele y \right\}. \]
\end{itemize}
\end{definition}

\paragraph{Semantics of Local Lexing}

Given a local lexing $\locallexing$, what is its semantics, i.e. how does it convert a sequence $D$ of characters into a set $\locallexing(D)$ of token sequences? Note that while for \emph{defining} a particular local lexing $\locallexing$ we do not need a context-free grammar, just its sets of terminals and characters, the \emph{semantics} of $\locallexing$ is dependent upon a grammar. In the following we will often call token sequences \emph{paths}.

The conversion process works as follows:
We go through $D$ from left to right, producing sets of paths $\paths k u \subseteq (\terminals \times \Sigma^*)^*$ along the way. The index $k \in \{0, \ldots, |D|\}$ denotes the position in $D$ we are looking at, and we use the index $u \in \{0, 1, \ldots\}$ to track iteratively generated $\paths k 0$, $\paths k 1$, $\ldots$ at this position. Let's assume now that we have reached position 
$k$ in $D$, having so far produced the set of paths $\paths k 0$; in case we
are at the very beginning of $D$ such that $k = 0$ we assume $\paths 0 0 = \{\emptyseq\}$. We will have been 
careful to only produce $p \in \paths k 0$ such that $\terminalsof{p} \in \prefixlang$. Now, which token should we produce next? For sure, the token to produce next must be a member of the set
\[ \tokensat k = \left\{ x \in \terminals \times \Sigma^* \mid x \in \lexer(\terminalsof{x})(D, k) \right\} \]
because according to $\lexer$ these are the only tokens which start at position $k$ in $D$.
Furthermore, we are only interested in the subset $\admissibletokens k 0$ of those members of $\tokensat k$ which can continue a sequence in 
$\paths k 0$ to a sequence of terminals in $\prefixlang$:
\[ \admissibletokens k 0 = \left\{ x \in \tokensat k\ |\ \exists\, p \in \paths k 0.\ |\charsof{p}| = k\ \wedge\ \terminalsof{p x} \in \prefixlang \right\}. \]
The selector decides whether and how to constrain any remaining ambiguity in the token selection: The set
\[ \selectedtokens k 1 = \selector(\emptyset, \admissibletokens k 0) \]
contains those tokens we use to create the next set of paths $\paths k 1$. 

If $\selectedtokens k 1$ does not contain any empty tokens, then $\paths k 1$ is just
\[\paths k 1  = \appendtokens k\ {\selectedtokens k 1}\, {\paths k 0},\]
where
\begin{align*}
&\appendtokens k\ T\ P = P\ \cup \\ 
&\quad \left\{ p t  \mid p \in P \wedge |\charsof{p}| = k \wedge t \in T \wedge \terminalsof{p t} \in \prefixlang \right\}, 
\end{align*}
and because we only added paths $q$ with $|\charsof{q}| > k$ we can simply proceed to position $k + 1$ via
$\paths {k+1} 0 = \paths {k} 1$. 

But if $\selectedtokens k 1$ does indeed contain empty tokens, then there might be newly added paths $q$ with $|\charsof{q}| = k$, and we might
be able to extend these paths even further. Therefore we define $\paths k 1$ more generally as
\[ \paths k 1  = \limit {(\appendtokens k\ {\selectedtokens k 1})}{\paths k 0}, \]
where 
\[\limit f X = \bigcup\limits_{n = 0}^\infty f^n(X).\]
The fact that now $\paths k 1 \setminus \paths k 0$ may contain paths $q$ with $|\charsof{q}| = k$ means that potentially more tokens in $\tokensat k$ 
become eligible to extends paths which stop at position $k$. We therefore keep repeating the above procedure (potentially infinitely often) until we are sure to have produced all eligible tokens at position $k$ by forming monotone chains 
\[
\begin{array}{ccccccc}
\admissibletokens k 0 &\subseteq& \admissibletokens k 1 &\subseteq& \admissibletokens k 2 &\subseteq& \ldots\\[0.8mm]
\rotatebox{90}{$\subseteq$} & & \rotatebox{90}{$\subseteq$} & & \rotatebox{90}{$\subseteq$} & & \\
\selectedtokens k 0 &\subseteq& \selectedtokens k 1 &\subseteq& \selectedtokens k 2 &\subseteq& \ldots\\[1mm]
\paths k 0 &\subseteq& \paths k 1 &\subseteq& \paths k 2 &\subseteq& \ldots
\end{array}
\]
where for $u \in \{0, 1, 2, \ldots\}$ we recursively define
\[
\begin{array}{rcl}
\admissibletokens k u &=& \left\{ x \in \tokensat k\ |\ \exists\, p \in \paths k u.\ |\charsof{p}| = k\ \wedge\ \terminalsof{p x} \in \prefixlang \right\}\\[2mm]
\selectedtokens k {u+1} &=& \selector(\selectedtokens k u, \admissibletokens k {u + 1})\\[2mm]
\paths k {u+1} &=& \limit {(\appendtokens k\ {\selectedtokens k {u+1}})}{\paths k u}.
\end{array}
\]
We finalize the generation of token sequences at position $k$ by defining 
\[\paths k \infty = \bigcup\limits_{u = 0}^\infty \paths k u \]
and move on to position $k+1$ via $\paths {k+1} 0 = \paths k \infty$. To complete the given set of recursive equations, we define $\selectedtokens k 0 = \emptyset$ 
for all positions $k$. 

Once we have arrived at the end of the character sequence $D$, we finish the conversion of $D$ into the set $\locallexing(D)$ of token sequences via defining
$\allpaths = \paths {|D|} \infty$ and
\[ \locallexing(D) = \left\{ p \in \allpaths \mid  |\charsof{p}| = |D| \wedge \terminalsof{p} \in \lang \right\}. \]

Note that the local lexing semantics does not depend on the particular form of the grammar, but only on $\lang$ and $\prefixlang$. If the grammar contains no unproductive nonterminals, i.e. if for all $X \in \nonterminals$ there is $w \in \Sigma^*$ such that $X \derives w$, then $\prefixlang$ is uniquely determined by $\lang$ and thus in this case the local lexing semantics only depends on $\lang$ alone.

Let us also comment on why the selector takes two arguments. It is tempting (and indeed this was the first thing we tried) to let the selector act only on a single argument, as in $\selectedtokens k {u+1} = \selector\ \admissibletokens k {u + 1}$, but this destroys the property that the $\selectedtokens k u$ form a monotone chain, which turned out to be important for the correctness of our algorithm for local lexing (see Section~\ref{sec:proof}). 

\section{Applications of Local Lexing}
\label{sec:applications}

In this section we explore local lexing and its applications through a range of examples.

\begin{example}[Traditional Lexical Specifications]
Turning a traditional lexical specification into the definition of a local lexing is straightforward. Let the traditional specification be defined over an input alphabet $\Sigma$ and given by $n$ pairs \[(r_1, t_1)\ \ldots\ (r_n, t_n).\] 
The $r_i$ are regular expressions over $\Sigma$, none of which match $\emptyseq \in \Sigma^*$, and the set of terminals $\terminals$ consists of $n$ different terminals 
$t_1, \ldots, t_n$.  
For each $t_i \in \terminals$ we then define \[\lexer(t_i)(D, k) = \{(t_i, c)\},\] where $c$ is the longest prefix of $D_k \ldots D_{|D|-1}$ such that 
$r_i$ matches $c$. If $r_i$ matches no prefix of $D_k \ldots D_{|D|-1}$ then we define $\lexer(t_i)(D, k) = \emptyset$. 

We define a strict partial order $\sele$ on the set of tokens which encodes the longest-match and priority rules:
\[ (t_i, c) \sele (t_j, d) \quad \text{iff} \quad |c| < |d| \vee (|c| = |d| \wedge i < j) .\]
This implies that $\selector(\emptyset, X)$ will be either empty or a singleton because $\sele$ is total on all possible $X$.

To endow the local lexing $\locallexing = (\lexer, \selector)$ with the equivalent semantics to traditional lexing, we define the context-free grammar
$G = (\{\nterminal{S}, \nterminal{T}\}, \terminals, \rules, \nterminal{S})$ where 
\begin{displaymath}
\rules = \left\{
\begin{array}{rcl}
\nterminal{S} & \rightarrow & \nterminal{S}\, T,\\
\nterminal{S} & \rightarrow & \emptyseq, \\
\nterminal{T} & \rightarrow & t_1,\\
& \vdots & \\
\nterminal{T} & \rightarrow & t_n
\end{array}
\right\}.
\end{displaymath}
Because of $\prefixlang = \lang = \terminals^*$ this grammar poses no additional constraints on the local lexing process, 
and therefore $\locallexing(D) = \emptyset$ iff traditional lexing of $D$ would lead to an error, and $\locallexing(D) = \{p\}$
iff traditional lexing of $D$ would yield the token sequence $p$. 
\end{example}

\begin{example}[Infinite Set of Token Sequences]\label{ex:lexinfinite}
We change the previous example slightly and allow the $r_i$ to also match the empty character sequence $\emptyseq$. As a concrete example
consider $\Sigma = \{\texttt{a}\}$, $n = 1$ and let $r_1$ be a regular expression which matches any (possibly empty) sequence of \texttt{a}'s.
Then using the same grammar $G$ as in the previous example, we obtain for example
\[ \locallexing(\texttt{aa}) = \left\{ \lextoken{aa}{$t_1$}, \lextoken{aa}{$t_1$} \lextoken{$\emptyseq$}{$t_1$},
 \lextoken{aa}{$t_1$} \lextoken{$\emptyseq$}{$t_1$} \lextoken{$\emptyseq$}{$t_1$}, 
 \lextoken{aa}{$t_1$} \lextoken{$\emptyseq$}{$t_1$} \lextoken{$\emptyseq$}{$t_1$} \lextoken{$\emptyseq$}{$t_1$}, 
 \ldots\right\},\]
which is an infinite set.
\end{example}

Because local lexing intertwines lexing and parsing, determining for a character sequence $D \in \Sigma^*$ all possible
token sequences $\locallexing(D)$ involves finding valid parses of $\terminalsof{p}$ for all $p \in \locallexing(D)$. The following examples
demonstrate this interplay between lexing and parsing and the role of the selector $\selector$ in it.

\begin{example}\label{ex:H}
Consider the grammar $H = (\nonterminals, \terminals, \rules, \startsymbol)$ where 
$\nonterminals = \{ \nterminal{S},\ \nterminal{A},\ \nterminal{E} \}$, 
$\terminals = \{ \terminal{plus},\ \terminal{minus},\ \terminal{id},\ \terminal{symbol} \}$,
$\rules = \{\nterminal{S} \rightarrow \nterminal{S}\ \terminal{plus}\ \nterminal{A},\ 
\nterminal{S} \rightarrow \nterminal{S}\ \terminal{minus}\ \nterminal{A},\ 
\nterminal{S} \rightarrow \nterminal{A},\ 
\nterminal{A} \rightarrow \nterminal{A}\ \nterminal{E},\ 
\nterminal{A} \rightarrow \nterminal{E},\ 
\nterminal{E} \rightarrow \terminal{id},\  
\nterminal{E} \rightarrow \terminal{symbol} \}$,  
$\startsymbol = \nterminal{S}$. 
The input characters are $\Sigma = \{\texttt{+}, \texttt{-}, \texttt{a}, \texttt{b}, \texttt{c}\}$.
As previously, we specify $\lexer$ simply by associating the terminals with regular expressions: The terminal \terminal{plus} 
recognizes the character \texttt{+}, \terminal{minus} recognizes the character \texttt{-}, \terminal{id} recognizes any nonempty sequence of letters and
\terminal{symbol} recognizes any nonempty sequence of letters and hyphens \texttt{-}. Let $D$ denote the character sequence \texttt{a-b+c}.

We choose the simplest option for $\sele$ and define \[\sele\ =\ \emptyset,\] i.e. no token has a higher priority than another token. This implies 
$\selector(A, B) = B$ for all token sets $A \subseteq B$. The result of determining $\locallexing(D)$ is shown in Figure~\ref{fig:lexH}. 
Because \terminal{symbol} overlaps with both \terminal{identifier} and \terminal{minus}, $D$ can be interpreted in eight different ways as a token sequence. Because of the longest-match rule for regular expressions, 
\[\lextoken{a}{symbol}\ \lextoken{-}{minus}\ \lextoken{b}{id}\ \lextoken{+}{plus}\ \lextoken{c}{id}\ \notin \locallexing(D),\]
excluded together with several other token sequences which would otherwise qualify.

For all $p \in \locallexing(D)$ the terminal sequence $\terminalsof{p}$ yields a valid parse tree, i.e. there
are 8 different ways to parse $D$ although $H$ by itself is an unambiguous grammar; 
all of them are jointly depicted in the parse graph in Figure~\ref{fig:parseH}. 
Ambiguities are depicted as dashed lines in the graph: whenever there are multiple ways to proceed with the derivation of a nonterminal, each alternative is connected with the nonterminal by a dashed line.
\end{example}

\begin{figure}
\[
\locallexing(\texttt{a-b+c}) = {\small\left\{ 
\begin{tabular}{l}
\lextoken{a}{id}\ \lextoken{-}{minus}\ \lextoken{b}{id}\ \lextoken{+}{plus}\ \lextoken{c}{id}\ , \\[0.3cm]
\lextoken{a}{id}\ \lextoken{-}{minus}\ \lextoken{b}{id}\ \lextoken{+}{plus}\ \lextoken{c}{symbol}\ ,\\[0.3cm]
\lextoken{a}{id}\ \lextoken{-}{minus}\ \lextoken{b}{symbol}\ \lextoken{+}{plus}\ \lextoken{c}{id}\ , \\[0.3cm]
\lextoken{a}{id}\ \lextoken{-}{minus}\ \lextoken{b}{symbol}\ \lextoken{+}{plus}\ \lextoken{c}{symbol}\ , \\[0.3cm]
\lextoken{a}{id}\ \lextoken{-b}{symbol}\ \lextoken{+}{plus}\ \lextoken{c}{id}\ , \\[0.3cm]
\lextoken{a}{id}\ \lextoken{-b}{symbol}\ \lextoken{+}{plus}\ \lextoken{c}{symbol}\ , \\[0.3cm]
\lextoken{a-b}{symbol}\ \lextoken{+}{plus}\ \lextoken{c}{id}\ , \\[0.3cm]
\lextoken{a-b}{symbol}\ \lextoken{+}{plus}\ \lextoken{c}{symbol}
\end{tabular}
\right\}}
\] 
\caption{Lexing \texttt{a-b+c} in Example~\ref{ex:H}}
\label{fig:lexH}
\end{figure}

\begin{figure}
\begin{center}
\includegraphics[scale=0.5]{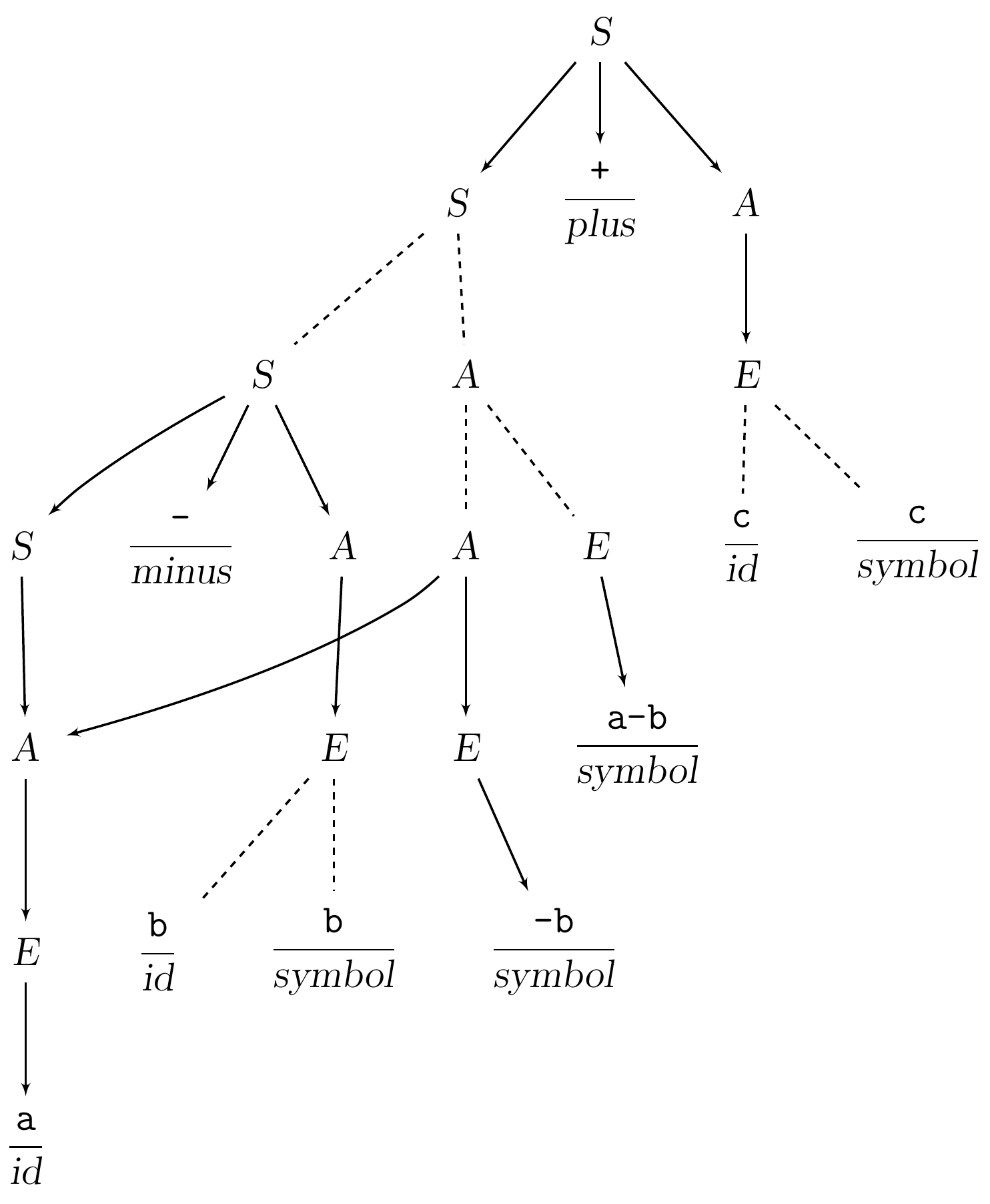}
\end{center}
\caption{Parsing \texttt{a-b+c} in Example~\ref{ex:H}}
\label{fig:parseH}
\end{figure}

\begin{example}\label{ex:H1}
We use the same grammer $H$ and the same lexer $\lexer$ as in Example~\ref{ex:H}, but choose a selector $\selector$ such that $\terminal{symbol}$
has lower priority than $\terminal{id}$ and $\terminal{minus}$:
\[ x \sele y \quad \text{iff} \quad [x] = \terminal{symbol} \wedge [y] \in \{\terminal{id}, \terminal{minus}\}. \]
This results in a unique lexing of \texttt{a-b+c},
\[\locallexing(\texttt{a-b+c}) = \left\{ \lextoken{a}{id}\ \lextoken{-}{minus}\ \lextoken{b}{id}\ \lextoken{+}{plus}\ \lextoken{c}{id} \right\},\]
and thus also in a unique parsing as shown in Figure~\ref{fig:H1}.
\end{example}
\begin{figure}
\begin{center}
\includegraphics[scale=0.5]{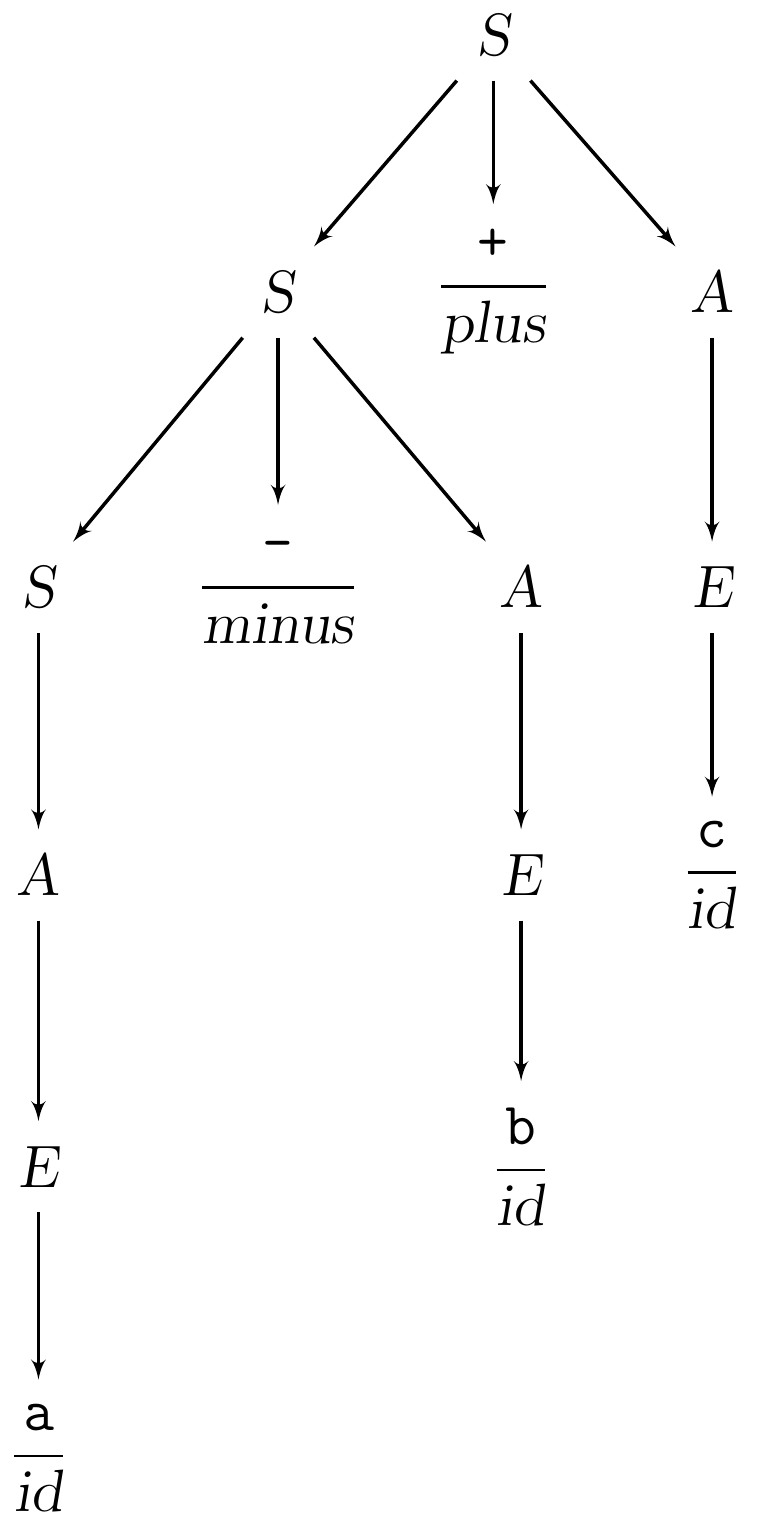}
\end{center}
\caption{Parsing \texttt{a-b+c} in Example~\ref{ex:H1}}
\label{fig:H1}
\end{figure}

\begin{example}\label{ex:H2}
We again leave $H$ and $\lexer$ fixed. This time we define $\sele$ such that longer tokens have higher priority:
\[ x \sele y \quad \text{iff} \quad |x| < |y|. \]
This yields two possible lexings for $\texttt{a-b+c}$,
\[ \locallexing(\texttt{a-b+c}) = 
\small\left\{ 
\begin{tabular}{l}
\lextoken{a-b}{symbol}\ \lextoken{+}{plus}\ \lextoken{c}{id}\ , \\[0.3cm]
\lextoken{a-b}{symbol}\ \lextoken{+}{plus}\ \lextoken{c}{symbol}
\end{tabular}
\right\},
\]
and leads to the ambiguous parsing shown on the left hand side of Figure~\ref{fig:H23}.
\end{example}

\begin{example}\label{ex:H3}
We choose to modify the previous example such that its remaining ambiguity is resolved in favour of $\terminal{id}$ by 
defining $x \sele y$ iff 
\[ |x| < |y| \vee (|x| = |y| \wedge [x] = \terminal{symbol} \wedge [y] = \terminal{id}). \]
This leads to a unique lexing for $\texttt{a-b+c}$,
\[ \locallexing(\texttt{a-b+c}) = 
\small\left\{ 
\lextoken{a-b}{symbol}\ \lextoken{+}{plus}\ \lextoken{c}{id}
\right\},
\]
and yields the unique parsing depicted on the right hand side of Figure~\ref{fig:H23}.
\end{example}

\begin{figure}
\begin{center}
\includegraphics[scale=0.5]{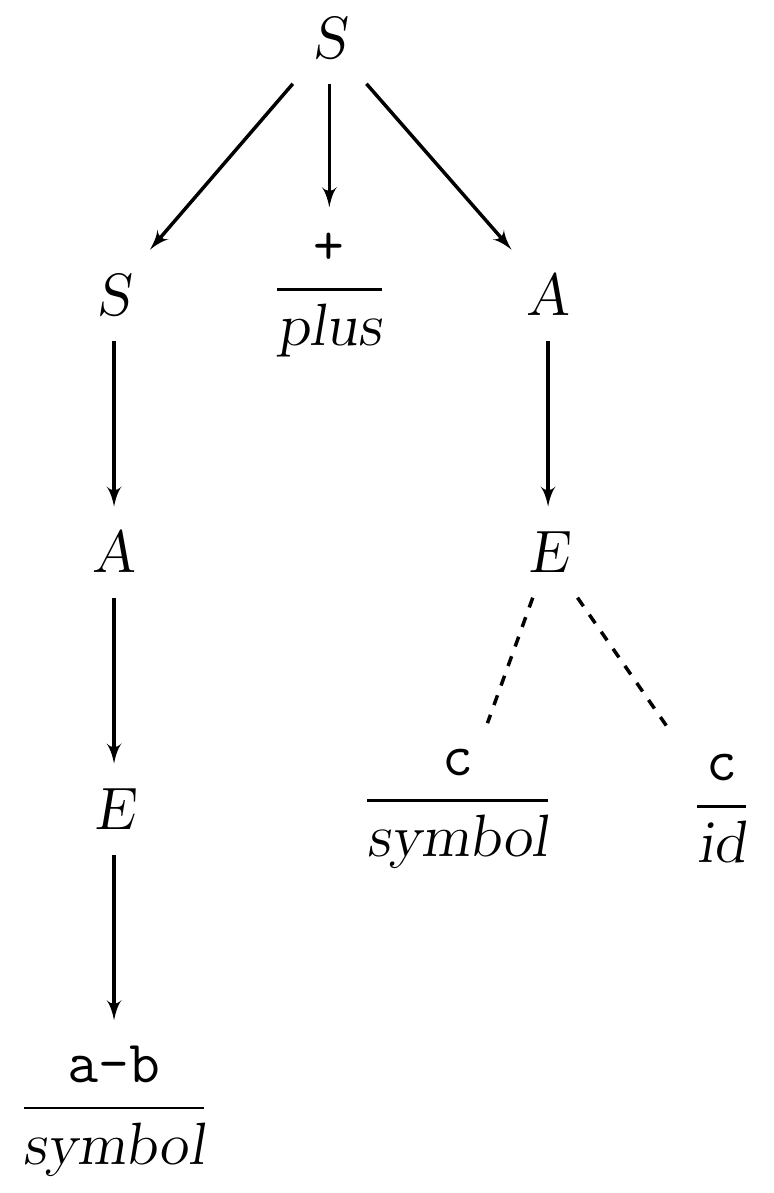}
\includegraphics[scale=0.5]{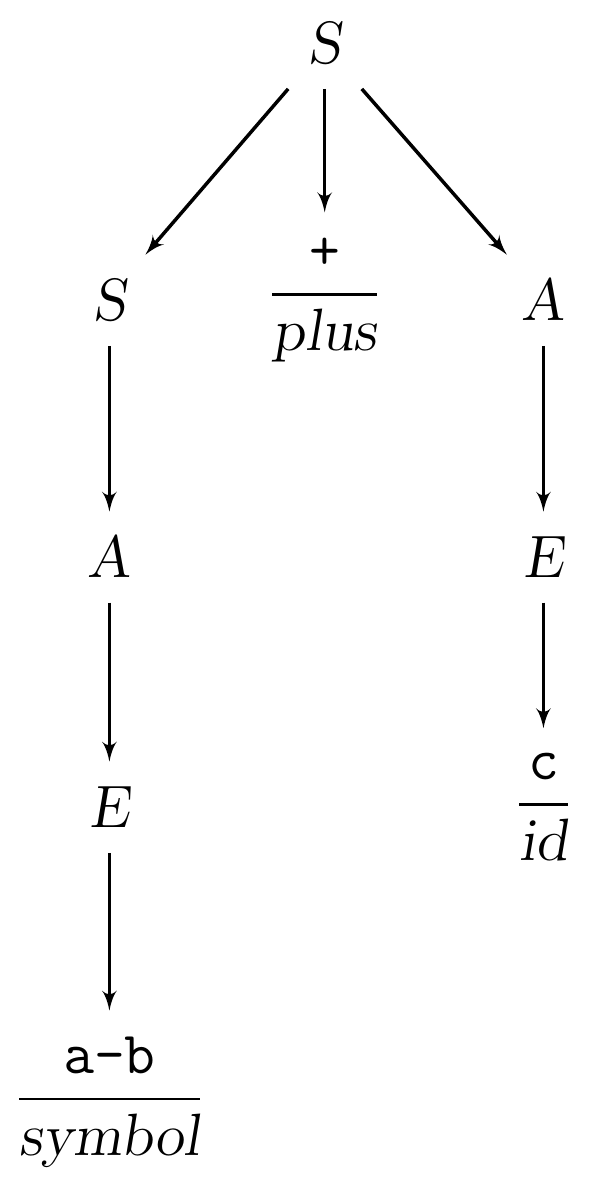}
\end{center}
\caption{Parsing \texttt{a-b+c} in Ex.~\ref{ex:H2} (left) and Ex.~\ref{ex:H3} (right)}
\label{fig:H23}
\end{figure}

\begin{example}[The Lexer Hack]\label{ex:lexerhack}
Consider the expression \texttt{(a)*b} given in the programming language C.
The meaning of this expression depends on whether \texttt{a} is a type identifier or a variable identifier. If it is a type identifier, then the expression is to be interpreted as a type cast, otherwise as a multiplication. Because type identifiers cannot be distinguished from variable identifiers by their look, in traditional parsing a problem arises, because the lexer phase is supposed to happen \emph{before} the parsing phase and any semantic analysis, but the terminal type of \texttt{a} really depends on information from the semantic analysis. The traditional solution, to manually direct feedback from the semantic analysis back into the lexer, is known as \emph{the lexer hack}~\cite{lexerhack}. With local lexing, the lexer / parser can offer \emph{both} alternatives, and let later stages of the analysis pick the right one, thus decoupling parsing and semantic analysis. The grammar $C$ demonstrates this, where 
$\terminals = \{\terminal{typeid},\terminal{id},\terminal{asterisk},\terminal{left},\terminal{right}\}$ and 
\begin{align*}
\rules = \left\{
\begin{array}{rcl}
\nterminal{Expr} & \rightarrow & \nterminal{Mul}, \\
\nterminal{Expr} & \rightarrow & \nterminal{Cast}, \\
\nterminal{Expr} & \rightarrow & \nterminal{Deref}, \\
\nterminal{Expr} & \rightarrow & \terminal{id}, \\
\nterminal{Expr} & \rightarrow & \terminal{left}\ \nterminal{Expr}\ \terminal{right}, \\
\nterminal{Mul} & \rightarrow & \nterminal{Expr}\ \terminal{asterisk}\ \nterminal{Expr}, \\
\nterminal{Cast} & \rightarrow & \terminal{left}\ \nterminal{Type}\ \terminal{right}\ \terminal{Expr}, \\
\nterminal{Deref} & \rightarrow & \terminal{asterisk}\ \nterminal{Expr}, \\
\nterminal{Type} & \rightarrow & \terminal{typeid}
\end{array}
\right\}.
\end{align*}
The lexer $\lexer$ is chosen in the obvious way, together with an empty $\sele$.
The resulting ambiguous parse graph for \texttt{(a)*b} is shown in Figure~\ref{fig:lexerhack}, and
\[ \locallexing(\texttt{(a)*b}) = 
\small\left\{ 
\begin{tabular}{l}
\lextoken{(}{left}\ \lextoken{a}{id}\ \lextoken{)}{right}\ \lextoken{*}{asterisk}\ \lextoken{b}{id}\ , \\[0.3cm]
\lextoken{(}{left}\ \lextoken{a}{typeid}\ \lextoken{)}{right}\ \lextoken{*}{asterisk}\ \lextoken{b}{id}\  \\[0.3cm]
\end{tabular}
\right\}.
\]
\end{example}
\begin{figure}
\begin{center}
\includegraphics[scale=0.5]{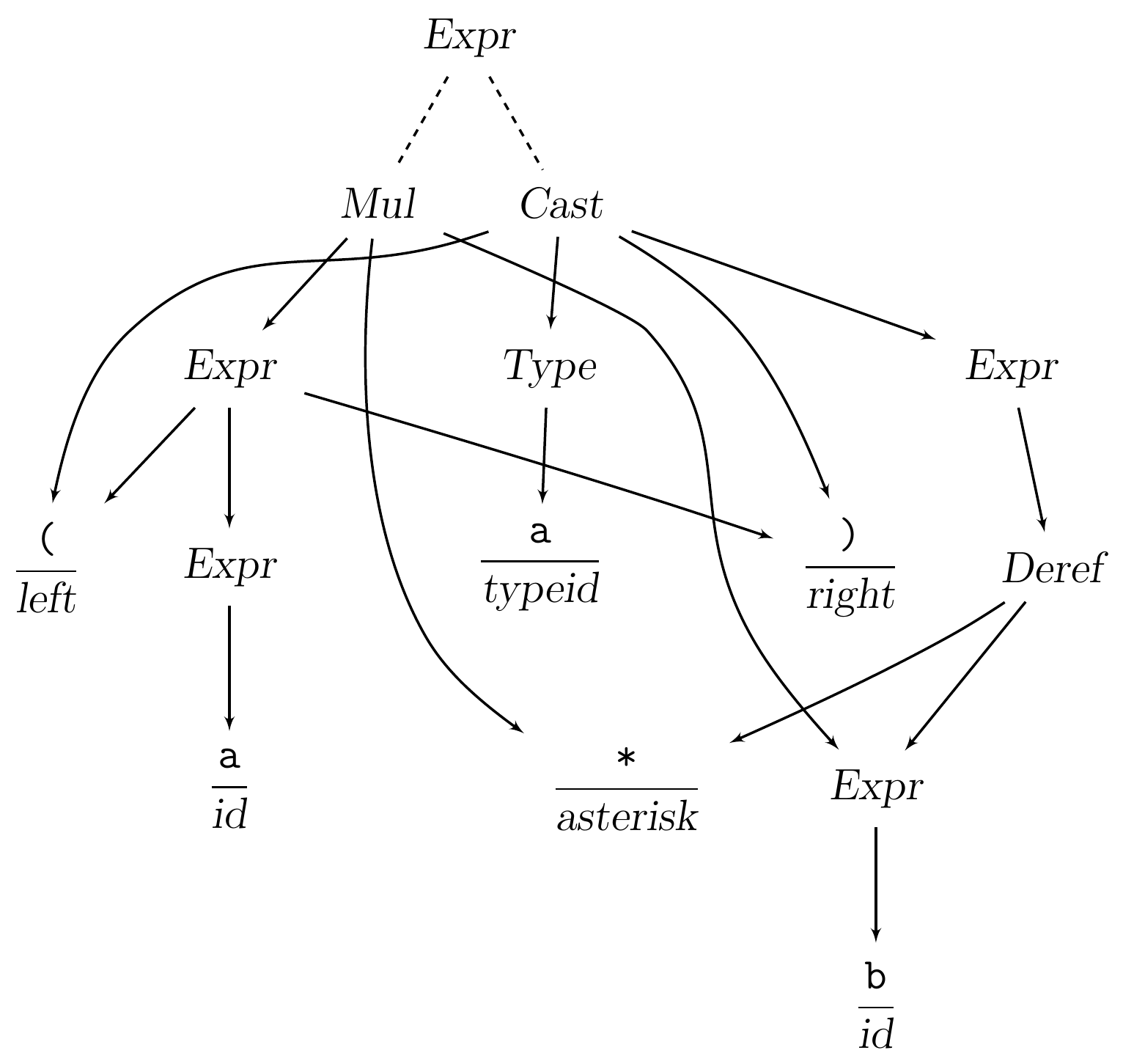}
\end{center}
\caption{Parsing \texttt{(a)*b} in $C$}
\label{fig:lexerhack}
\end{figure}

\paragraph{Schr\"odinger's Token}
Example~\ref{ex:lexerhack} is a special case of a scenario where the conversion from character sequences to token sequences is ambiguous, but where the token \emph{boundaries} are always the same in all alternative token sequences. The tokens appearing in such a scenario have been christened Schr\"odinger's tokens~\cite{schroedingerstoken}. Local lexing is more powerful than the Schr\"odinger's token approach and (at least conceptually) subsumes it.  

\paragraph{Ruby Slippers}
The Marpa parser~\cite{marpa} is an Earley-based parsing library which advocates the use of a technique called \emph{Ruby Slippers}~\cite{rubyslippers}.
This technique takes advantage of the fact that the Earley parser "knows" which tokens it expects at any given stage of the parse progress. Marpa has an interface through which the parser can communicate with the scanner to negotiate which token to scan next, thus allowing for sophisticated error handling.

Ruby Slippers and local lexing are both children of the same insight, namely that the scanning and parsing stages should communicate because the parser has information about which tokens it expects next. Local lexing though takes this insight to a new level which in principle is independent from a particular parsing algorithm like Earley. In this sense, local lexing can be seen as providing a rigorous semantics for certain uses of the Ruby Slippers technique. 

The next example demonstrates how local lexing can be used to specify lightweight error recovery as part of the language design.
\begin{example}[Error Recovery]\label{ex:errorrecovery} 
Consider a grammar for simple arithmetic expressions, where the nonterminals are given by
$\{ \nterminal{Expr},\ \nterminal{Sum},\ \nterminal{Mul},\ \nterminal{Atom} \}$, the terminals by 
$\{ \terminal{plus},\ \terminal{mul},\ \terminal{id},\ \terminal{num}, \ \terminal{left}, \ \terminal{right} \}$, and the rules by
\begin{align*}
\begin{array}{rcl}
\nterminal{Expr} & \rightarrow &\nterminal{Sum},\\ 
\nterminal{Sum} & \rightarrow & \nterminal{Sum}\ \terminal{plus}\ \nterminal{Mul},\\ 
\nterminal{Sum} & \rightarrow & \nterminal{Mul},\\ 
\nterminal{Mul} & \rightarrow & \nterminal{Mul}\ \terminal{mul}\ \nterminal{Atom},\\ 
\nterminal{Mul} & \rightarrow & \nterminal{Atom},\\ 
\nterminal{Atom} & \rightarrow & \terminal{left}\ \terminal{Sum}\ \terminal{right},\\  
\nterminal{Atom} & \rightarrow & \terminal{id},\\  
\nterminal{Atom} & \rightarrow & \terminal{num}.
\end{array}
\end{align*}
The input characters are $\Sigma = \{\texttt{+}, \texttt{*}, \texttt{(}, \texttt{)}, \texttt{0} \dots \texttt{9}, \texttt{a} \dots \texttt{z}\}$,
and $\lexer$ is specified such that \terminal{plus} recognizes the character \texttt{+}, \terminal{mul} recognizes the character \texttt{*}, 
\terminal{left} recognizes the opening bracket \texttt{(}, \terminal{right} recognizes the closing bracket \texttt{)}, \terminal{id} recognizes any nonempty sequence of letters and digits starting with a letter, and \terminal{num} recognizes any nonempty sequence of digits. Consider now the following string $w$ which is invalid with respect to the language we just specified:
\[\texttt{2(a*+))+(1}\] 
Instead of simply diagnosing that there is some parse error at position $k=1$, we would like to recover some of the structure of $w$, for example for providing better error messages or for a rich interactive editing experience. To simplify this task, we first make our grammar more permissive by adding the rule
$\nterminal{Mul}\ \rightarrow\ \nterminal{Mul}\ \nterminal{Atom}$.
This has the effect that a string like \texttt{2x} becomes legal, representing the multiplication of \texttt{2} and \texttt{x}. Such notation is common mathematical practice and thus seems like a justifiable design choice. We then introduce three new terminals: 
\terminal{e-atom}\ and \terminal{e-right}\ both recognize the empty string $\varepsilon$ only, and \terminal{e-superfluous}\ recognizes the closing bracket \texttt{)}.
We incorporate these new terminals into the grammar by adding the following rules:
\begin{align*}
\begin{array}{rcl}
\nterminal{Mul} & \rightarrow & \nterminal{Mul}\ \terminal{e-superfluous},\\ 
\nterminal{Atom} & \rightarrow & \terminal{left}\ \terminal{Sum}\ \terminal{e-right},\\  
\nterminal{Atom} & \rightarrow & \terminal{e-atom}  
\end{array}
\end{align*}
Finally, we choose the selector $\selector$ such that the three added error terminals have \emph{lower priority than all other terminals}.
Paths in $\locallexing(D)$ that contain error terminals will only be considered by us if $\locallexing(D)$ contains no paths without error terminals. Figure~\ref{fig:errorrecovery} shows the result of applying the updated grammar to $w$. The corresponding path is 
% num["2"] left["("] id["a"] mul["*"] e-atom[""] plus["+"] e-atom[""] right[")"] e-superfluous[")"] plus["+"] left["("] num["1"] e-right[""]
\begin{multline*}
\lextoken{2}{num}\ \lextoken{(}{left}\ \lextoken{a}{id}\ \lextoken{*}{mul}\ \lextoken{$\varepsilon$}{e-atom}\ \lextoken{+}{plus}\ \lextoken{$\varepsilon$}{e-atom}\ \lextoken{)}{right}\\ \lextoken{)}{e-superfluous}\ \lextoken{+}{plus}\ 
\lextoken{(}{left}\ \lextoken{1}{num}\ \lextoken{$\varepsilon$}{e-right}.
\end{multline*}
Testing convinces us that the updated grammar can indeed successfully parse all $D \in \Sigma^*$, and does so unambiguously.

\begin{figure}
\begin{center}
\includegraphics{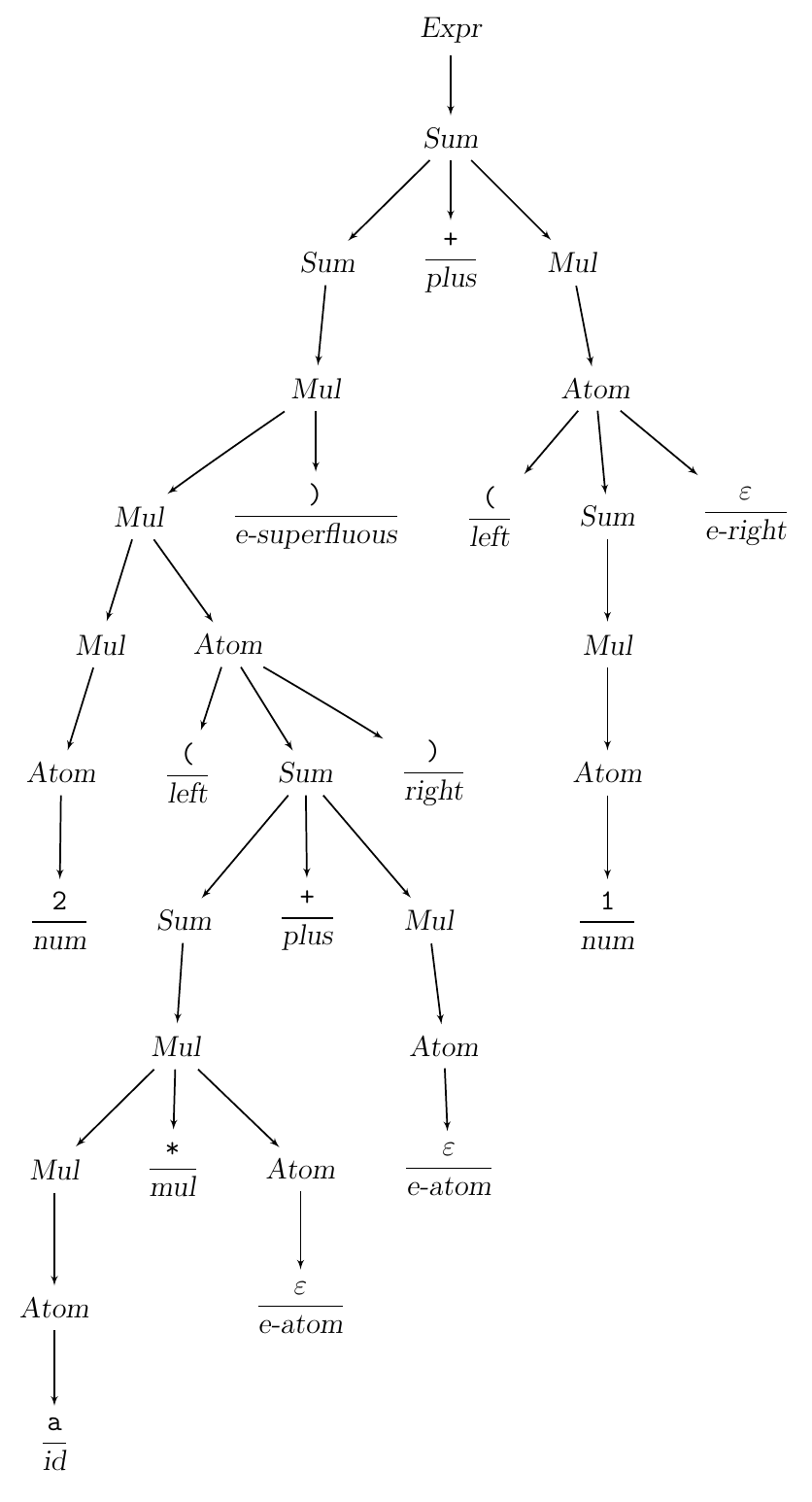}
\end{center}
\caption{Parsing \texttt{2(a*+))+(1} in Example~\ref{ex:errorrecovery}}
\label{fig:errorrecovery}
\end{figure}

\end{example}

\paragraph{Whitespace}
In our examples we have avoided to make use of whitespace. In traditional parsing there are essentially two different ways of dealing with whitespace:
\begin{enumerate}
\item One approach is to explicitly incorporate whitespace terminals into the grammar rules. This can become cumbersome and error-prone if done manually, because usually whitespace can legally appear almost everywhere.
\item The other approach is to handle whitespace at the lexical stage exclusively. 
\end{enumerate}
The first approach applies to local lexing as well. The second approach does not directly apply, as there is no single lexical stage anymore with local lexing. 
Nevertheless, it seems that local lexing allows to combine the convenience of the second approach with the fine-grained control of the first one by using extended attribute grammars to specify layout constraints, making it possible to tackle layout-sensitive languages. The combination of local lexing with layout-sensitivity is work in progress and beyond the scope of this paper.

\section{Implementing Local Lexing}
\label{sec:algorithm}

\newcommand{\sigmalang}{\ensuremath{\lang_\Sigma}}
\newcommand{\bin}{\ensuremath{\textsl{Bin}}}
\newcommand{\items}{\ensuremath{\mathcal{I}}}
\newcommand{\validitems}{\ensuremath{\mathcal{V}}}
\newcommand{\itemdot}{\ensuremath{\makebox[7pt]{\textbullet}}}
\newcommand{\someitems}{\ensuremath{I}}
\newcommand{\generateditems}[1]{\ensuremath{\langle{#1}\rangle}}
\newcommand{\itemsi}[1]{\ensuremath{\mathcal{I}_{#1}}}
\newcommand{\itemsj}[2]{\ensuremath{\mathcal{J}_{#1}^{#2}}}
\newcommand{\ctokens}[2]{\ensuremath{\mathcal{T}_{#1}^{#2}}}

Given a grammar $G$ and a local lexing $\locallexing$, let us define the \emph{character language} $\sigmalang$ of $G$ and $\locallexing$ by
\[ \sigmalang = \left\{ D \in \Sigma^* \mid \locallexing(D) \neq \emptyset \right\}. \]
How do we build a recognizer for $\sigmalang$?

\paragraph{Lexing Driving Parsing}
Our first attempt might be to directly apply the semantics of local lexing after 
having picked a parsing algorithm which is capable of recognizing both the prefixes $\prefixlang$ and the language $\lang$ of $G$. 
Most of the popular parsing algorithms would be suitable for this, such as LL, LR or Earley parsing.

While this approach will work in many cases, it is inefficient to treat the parsing algorithm as a black box which is repeatedly asked whether
a given sequence of terminals is in $\prefixlang$ or not. More importantly, as Example~\ref{ex:lexinfinite} shows there are finite grammars which nevertheless produce
infinite sets of token sequences and for which this approach would therefore fail by getting stuck in a non-terminating path generating loop.  

\paragraph{Parsing Driving Lexing} 
A better approach seems to be to inverse above approach and to let the parsing progress drive the lexing process. Often it will be possible to predict from the internal parser state which terminal is expected next. For this to work, we need to modify
the parsing algorithm so that it not only knows about terminals $\terminals$, but also about characters $\Sigma$.
The Earley algorithm seems to be best suited to be adapted to such a purpose, as it works on all context-free grammars, copes gracefully with ambiguity, and is easily extended with top-down, left-right, and bottom-up parametricity. This is why Earley-based parsing is our main focus here. Nevertheless, studying how to modify other parsing algorithms for local lexing is interesting and of potentially great practical interest as well; experiments indicate that in particular the LR(1) parsing algorithm can be modified to facilitate local lexing in a natural and simple way. 

\paragraph{Earley's Algorithm} We first describe (a high-level version of) Earley's original algorithm, assuming $\Sigma = \terminals$. 
To recognize an input $D \in \Sigma^*$ as belonging to $\lang$, it computes \emph{items}; an item is a quadruple $(r, d, i, j)$ where 
$r = (N \rightarrow \alpha\beta) \in \rules$ is a rule of the grammar, $d = |\alpha|$ is a position within that rule demarking the current parsing progress, $i \in \{0, \ldots, |D|\}$ is the origin of the item, and $j \in \{i, \ldots, |D|\}$ is the bin of the item. 
An alternative way to write the item is as \[(N \rightarrow \alpha\itemdot\beta, i, j).\] 

Earley's algorithm builds a monotone chain of item sets
\[ \itemsi 0 \subseteq \itemsi 1 \subseteq \itemsi 2 \subseteq \ldots \subseteq \itemsi {|D|} = \allitems. \]
To this end, we define an initial item set $\operatorname{Init}$ and 
monotone operators $\operatorname{Predict}$, $\operatorname{Complete}$ and $\operatorname{Scan}$ which all take a position 
$k \in \{0, \ldots, |D|\}$ and an item set $\someitems$ and return an augmented item set (Figure~\ref{fig:earleyoperators}).  
\begin{figure}
\begin{align*}
\operatorname{Init} & = \{ (\startsymbol \rightarrow \itemdot \alpha, 0, 0) \mid \startsymbol \rightarrow \alpha \in \rules\} \\[0.2cm]
\operatorname{Predict}\ & k\ \someitems = \someitems\ \cup\\
 \{ & (M \rightarrow \itemdot \gamma, k, k) \mid \exists\,N\ \alpha\ \beta\  i.\\
& (N \rightarrow \alpha \itemdot M \beta, i, k) \in \someitems \wedge (M \rightarrow \gamma) \in \rules \}\\[0.2cm]
\operatorname{Complete}\ &  k\ \someitems = \someitems\ \cup\\
 \{ & (N \rightarrow \alpha M \itemdot \beta, i, k) \mid \exists\,j\ \gamma.\\
    & (N \rightarrow \alpha \itemdot M \beta, i, j) \in \someitems \wedge (M \rightarrow \gamma \itemdot, j, k) \in \someitems \} \\[0.2cm]
\operatorname{Scan}\ &  k\ \someitems = \someitems\ \cup\\
 \{ & (N \rightarrow \alpha X \itemdot \beta, i, k + 1) \mid k < |D| \wedge X = D_k\ \wedge  \\
     &(N \rightarrow \alpha \itemdot X \beta, i, k) \in \someitems \}
\end{align*}
\caption{Building Blocks of Earley's Algorithm}
\label{fig:earleyoperators}
\end{figure}
We then define for $k \in \{0, \ldots, |D|\}$ the operator $\pi_k$ by
\[ \pi_k\,\someitems = \limit {(\operatorname{Scan} k \circ \operatorname{Complete} k \circ \operatorname{Predict} k)}{\someitems},\]
reusing the $\operatorname{limit}$ operator introduced in Section~\ref{sec:lldef}, and use $\pi_k$ to recursively define the sets $\itemsi k$ by
\begin{align*} 
\itemsi 0 & = \pi_0\ \operatorname{Init},  \\
\itemsi k & = \pi_k\ \itemsi {k - 1}  \ \text{for $k > 0$}.
\end{align*}

\begin{theorem}[Correctness of Earley's Algorithm]\label{th:earleycorrectness}
Earley's algorithm is both sound and complete, i.e.
\[\startsymbol \derives D \quad\text{iff}\quad \exists\, \alpha.\ (\startsymbol \rightarrow \alpha \itemdot, 0, |D|) \in \allitems.\] 
\end{theorem}
\begin{proof}
This is covered by Theorem~\ref{th:llearleycorrectness} for the special case 
\begin{align*}
\Sigma & = \terminals, \\
\lexer(t)(D, k) & = 
  \begin{cases} \{(t, t)\} & \text{for $k < |D| \wedge D_k = t$} \\ \emptyset & \text{otherwise} \end{cases}, \\ 
\sele &= \emptyset.
\end{align*}
\end{proof}

\paragraph{Earley's Algorithm with Local Lexing}
We now assume that we have a local lexing $\locallexing$ and thus $\Sigma$ and $\terminals$ do not necessarily coincide anymore.
We leave $\operatorname{Init}$, $\operatorname{Predict}$ and $\operatorname{Complete}$ unchanged, but we need an updated $\operatorname{Scan}$ operator
that works on tokens instead of characters, and a new operator $\operatorname{Tokens}$ (Figure~\ref{fig:anewscanner}).
\begin{figure}
\begin{align*}
\operatorname{Tokens}\ &  T\ k\ \someitems = \\
 \selector\ T\ \{ & x \mid \exists\,X\ N\ \alpha\ \beta\ i.\ X \in \terminals\ \wedge\ \\
     & (N \rightarrow \alpha \itemdot X \beta, i, k) \in \someitems\ \wedge x \in \lexer(X)(D, k)\}\\[0.2cm]
\operatorname{Scan}\ &  T\ k\ \someitems = \someitems\ \cup\\
 \{ & (N \rightarrow \alpha X \itemdot \beta, i, k + |c|) \mid (X, c) \in T\ \wedge\\
    & (N \rightarrow \alpha \itemdot X \beta, i, k) \in I \}
\end{align*}
\caption{A New Scanner}
\label{fig:anewscanner}
\end{figure}

The operation $\operatorname{Tokens} T\ k\ \someitems$ first determines all the candidate terminals
that could possibly appear next at position $k$ in $D$ according to $\someitems$. It then determines which of those candidate terminals
can actually be lexed as tokens at that position. It applies the selector $\selector$ to it and returns the resulting set of tokens. Finally,
the $\operatorname{Scan}$ operator is easily adapted to work on tokens instead of single characters. We iteratively compute the sets 
$\ctokens k 0$, $\ctokens k 1$, $\ldots$ of tokens at position $k$, and these act as arguments $T$ to both $\operatorname{Tokens}$ and $\operatorname{Scan}$.
Accordingly, we need to update the definition of $\pi_k$ to take the additional argument $T$ into account:
\[ \pi_k\,T\,\someitems = \limit {(\operatorname{Scan} T\ k \circ \operatorname{Complete} k \circ \operatorname{Predict} k)}{\someitems}.\]
Furthermore, it is now possible that scanning at position $k$ might add new items to bin $k$ due to the existence of empty tokens, 
therefore enlarging the set of eligible terminals at position $k$. To cope with this we keep applying the operator $\pi_k$ with updated token sets until it converges:
\begin{align*}
\itemsj 0 0 & = \pi_0\ \emptyset\ \operatorname{Init}\\
\itemsj k {u+1} & = \pi_k\ \ctokens k {u+1}\ \itemsj k u\\
\itemsi k & = \bigcup\limits_{u=0}^\infty \itemsj k u\\
\itemsj {k+1} 0 & = \pi_{k+1}\ \emptyset{}\ \itemsi k\\
\ctokens k 0 & = \emptyset\\
\ctokens {k} {u+1} & = \operatorname{Tokens}\ \ctokens k u \  k\ \itemsj k u.
\end{align*}
Note that in case of $\itemsj k {u+1} = \itemsj k u$ we have $\itemsi k = \itemsj k u$, thus the computation of $\itemsi k$ can stop at that point.
In particular, if $\ctokens k 1$ does not contain any empty tokens, then the computation of $\itemsi k$ simplifies to $\itemsi k = \itemsj k 1$.

Above equations for computing $\allitems = \itemsi {|D|}$ show an obvious correspondence to the equations we used for defining $\allpaths$ in 
Section~\ref{sec:lldef}. And indeed, Earley's algorithm with local lexing is correct with respect to the local lexing semantics:

\begin{theorem}[Correctness of Earley's Algorithm with Local Lexing]\label{th:llearleycorrectness}
Earley's algorithm with local lexing is both sound and complete, i.e.
\[D \in \sigmalang \quad\text{iff}\quad \exists\, \alpha.\ (\startsymbol \rightarrow \alpha \itemdot, 0, |D|) \in \allitems.\] 
\end{theorem}

\begin{proof}See Section~\ref{sec:proof}.\end{proof}

The above correctness result comes with a caveat: It may be the case that the computation requires an infinite amount of time and space. 
If the grammar is finite though, then for an input $D \in \Sigma^*$ the size of $\allitems$ is bounded by 
\[ \left({{|D| + 1}\choose 2} + |D| + 1\right) \sum\limits_{N \rightarrow \alpha \in \rules} 1 + |\alpha|,\] 
and thus the computation will require only a finite amount of time and space (assuming $\selector$ and $\lexer$ require only a finite amount in the first place).

\paragraph{Practical Implementation}
We have developed a practical library for local lexing written in Scala/Scala.js which can be used to try out the examples in Section~\ref{sec:applications}~\cite{locallexingprototype}. It is a fairly naive proof-of-concept implementation and not optimized for data structures at all. For future versions of the library we plan to examine which of the many established ideas for making Earley parsing faster also apply (at least partially) to our case. 

\section{Proof of Theorem~\ref{th:llearleycorrectness}}
\label{sec:proof}
In this section we give a short outline of the proof of Theorem~\ref{th:llearleycorrectness}. The full formal proof is available as Isabelle/HOL 2016 theory files and has been fully machine-checked for correctness~\cite{locallexingtheories}. To convince yourself that the formal proof really proves Theorem~\ref{th:llearleycorrectness} we recommend first studying theories \texttt{CFG}, \texttt{LocalLexing} and \texttt{LLEarleyParsing}. These contain 
the basic definitions in (almost) the same notation as presented in this paper. You should then proceed to look at theory \texttt{MainTheorems}. 
It contains the theorem \texttt{Correctness} which is the machine-checked counterpart of Theorem~\ref{th:llearleycorrectness}. 

\paragraph{Proof Outline} 
There is an intuitive correspondence between the sets $\paths k u$ and the sets $\itemsj k u$.  
Clarifying this correspondence is the most important step towards proving the correctness of the algorithm.
\begin{definition}[Valid and Generated Items]\label{def:validity}
We call an item \[(N \rightarrow \alpha \itemdot \beta, i, j) \in \allitems\] 
\emph{$p$-valid} for some token sequence $p \in \allpaths$ 
iff there is $u \in \{0, \ldots, |p|\}$ such that
\begin{align*}
  |\charsof{p}| &= j, \\
  |\charsof{p_0 \ldots p_{u-1}}| &= i,\\
  \startsymbol & \derives \terminalsof{p_0 \ldots p_{u-1}} N \gamma \quad\text{for some $\gamma$},\\
  \alpha &\derives \terminalsof{p_u \ldots p_{|p|-1}}.
\end{align*}
For $P \subseteq \allpaths$ we say that $P$ \emph{generates} $\generateditems{P}$, where
\[ \generateditems{P} = \left\{ x \in \allitems \mid \exists\, p \in P.\ \text{$x$ is $p$-valid}\right\}. \]
\end{definition}
This notion of validity has been inspired by the one introduced in \cite{jones-earley} where it has been defined as an absolute property of an item.
To make it work in our context, we had to define it not absolutely, but relatively with respect to a path.

The bulk of the proof consists then in proving $\allitems = \generateditems{\allpaths}$. We will not delve into the rather technical proof of this here, but we want to point out the following supporting theorem about paths: 
\begin{theorem}\label{th:generalizedcompatibility}
For all inputs $D$, for all $k \in \{0, \ldots, |D|\}$ and 
for all $u \in \{0, 1, 2, \ldots\}$ the following holds: Given $p, q \in \paths k u$ such that
$|\charsof{p}| = |\charsof{q_0 \ldots q_{n-1}}| \leq k$ for some $n \in \{0, \ldots, |q|\}$ and $[p\, q_n \ldots q_{|q|-1}] \in \prefixlang$, 
it follows that $p\, q_n \ldots q_{|q|-1} \in \paths k u$.
\end{theorem}
Intuitively this means that when there are two paths $p \in \allpaths$ and $q \in \allpaths$ which meet at some position $k$, i.e. $p = a b$ and $q = c d$ with
$|\charsof a| = |\charsof c| = k$, then they can crossover, i.e. both $a d$ and $c b$ will be in $\allpaths$ as long as this makes sense with respect to the grammar.
But $p$ and $q$ are not guaranteed to arrive at the same time $u$ at position $k$, and so it might be that tokens that were around when $p$ arrived are not there anymore when $q$ arrives, and vice versa. The fact that the token sets $\selectedtokens k u$ form a monotone chain at position $k$ means that this cannot happen. 

The semantics of local lexing is defined by mutually recursive equations which intertwine lexing and parsing, 
but once all token sets $\selectedtokens{k}{\infty} = \bigcup_{u = 0}^{\infty} \selectedtokens{k}{u}$
have been established it is possible to disentangle lexical and grammatical matters again as another supporting theorem about paths shows:
\begin{theorem}\label{th:disentanglement}
Let $p$ be a sequence of tokens. Then $p \in \allpaths$ iff 
\[ \text{a)}\ \displaystyle\mathop{\forall}_{0 \leq i < |p|} p_i \in \selectedtokens{|\charsof{p_0 \ldots p_{i-1}}|}{\infty} \quad 
\text{and}\quad \text{b)}\ \terminalsof{p} \in \prefixlang .\]
\end{theorem}

From $\allitems = \generateditems{\allpaths}$ and with the help of Theorem~\ref{th:disentanglement} it is then straightforward to prove Theorem~\ref{th:llearleycorrectness}:
\begin{proof}
Let us first assume $D \in \sigmalang$. This implies $\locallexing(D) \neq \emptyset$, which implies that there is a $p \in \allpaths$
with $|\charsof{p}| = |D|$ and $\startsymbol \derives \terminalsof{p}$. This means there is an $\alpha$ such that $\startsymbol \rightarrow \alpha$
and $\alpha \derives \terminalsof{p}$. Therefore, $(\startsymbol \rightarrow \alpha \itemdot, 0, |D|)$ is $p$-valid and thus
\[ (\startsymbol \rightarrow \alpha \itemdot, 0, |D|) \in \generateditems{\allpaths} = \allitems. \]

On the other hand, let us assume that there is an $\alpha$ with above property. This means that
$(\startsymbol \rightarrow \alpha \itemdot, 0, |D|)$ is $p$-valid for some $p \in \allpaths$, i.e. there is 
$u \in \{0, \ldots, |p|\}$ with
\begin{align*}
  |\charsof{p}| &= |D|, \\
  |\charsof{p_0 \ldots p_{u-1}}| &= 0,\\
  \startsymbol & \derives \terminalsof{p_0 \ldots p_{u-1}} \startsymbol \gamma \quad\text{for some $\gamma$},\\
  \alpha &\derives \terminalsof{p_u \ldots p_{|p|-1}}.
\end{align*} 
Above facts together with Theorem~\ref{th:disentanglement} show that by dropping the first $u$ empty tokens from $p$ we obtain
a path \[p_u \ldots p_{|p|-1} \in \locallexing(D).\]
\end{proof}

\section{Further Related Work}
\label{sec:relatedwork}
A strong influence on our work has been the idea of \emph{blackbox} Earley parsing presented in~\cite{blackbox}. A blackbox is a (possibly third-party) parser component plugged into the Earley parsing framework by associating the blackbox with a nonterminal. Our \lexer\ component can basically be viewed as a collection
of blackboxes, but instead of associating them with nonterminals, we associate them with terminals. This makes it possible to treat blackboxes as a concept that is in principle independent from Earley parsing. Unlike the original work on blackboxes we also provide a method for disambiguation, via the selector \selector.

\section{Conclusion}
\label{sec:conclusion}

With hindsight local lexing is a simple concept, but it has taken us over two years to arrive at the concept as it is presented in this paper. Our algorithm for local lexing and the semantics of local lexing developed side-by-side during this time. There have been enough missteps along the way to finally make us formally verify our algorithm. Despite its simplicity, the examples from Section~\ref{sec:applications} show that local lexing is a versatile and unifying concept for designing syntax. We hope that you may find it useful, too. 

\bibliographystyle{abbrvnat}

\begin{thebibliography}{99}
\softraggedright
%\bibitem{proofpeer-website} ProofPeer, \mbox{\url{http://www.proofpeer.net}}
%\bibitem{proofpeer-positionpaper} Steven Obua, Jacques Fleuriot, Phil Scott, David Aspinall: 
%ProofPeer: Collaborative Theorem Proving, \mbox{\url{http://arxiv.org/abs/1404.6186}} 
%\bibitem{flyspeck}Thomas Hales et al.: A formal proof of the Kepler conjecture,\hfill\\
%\mbox{\url{http://arxiv.org/abs/1501.02155}}
%\bibitem{oddorder}Georges Gonthier et al.: A Machine-checked Proof of the Odd Order theorem
%\bibitem{sel4}Gerwin Klein et al.: seL4: Formal verification of an OS kernel
%\bibitem{compcert}Xavier Leroy: Formal verification of a realistic compiler
\bibitem{blackbox}Trevor Jim, Yitzhak Mandelbaum, David Walker. \emph{Semantics and Algorithms for Data-dependent Grammars}, POPL 2010.
\bibitem{peg}Brian Ford. \emph{Parsing Expression Grammars: A Recognition Based Syntactic Foundation}, POPL 2004.
\bibitem{scannerless}Eelco Visser. \emph{Scannerless Generalized-LR Parsing}, Report P9707, University of Amsterdam 1997.
\bibitem{lexerhack}{The lexer hack,\flushleft} \mbox{\url{https://en.wikipedia.org/wiki/The_lexer_hack}}.
\bibitem{schroedingerstoken}John Aycock,  R. Nigel Horspool. \emph{Schr\"odinger's token}, Software Practice and Experience, 2001.
%\bibitem{erdweg}Sebastian Erdweg, Tillmann Rendel, Christian K\"astner, Klaus Ostermann: Layout-Sensitive Generalized Parsing
%\bibitem{adams}Michael D. Adams: Principled Parsing for Indendation-Sensitive Languages
%\bibitem{parsingschemata}Klass Sikkel: Parsing Schemata and Correctness of Parsing Algorithms
%\bibitem{okasaki-layout}Chris Okasaki: In praise of mandatory indentation for novice programmers, 
%\url{http://okasaki.blogspot.de/2008/02/in-praise-of-mandatory-indentation-for.html}, blog post, February 2008.
%\bibitem{landin}P. J. Landin: The Next 700 Programming Languages
\bibitem{isabelle}\emph{Isabelle}, \url{http://isabelle.in.tum.de}.
\bibitem{scala}\emph{Scala}, \url{http://scala-lang.org}.
\bibitem{scalajs}S\'{e}bastien Doeraene. \emph{Scala.js}, \url{http://scala-js.org}.
%\bibitem{markdown}John Gruber: Markdown, \url{https://daringfireball.net/projects/markdown}
\bibitem{jones-earley}Cliff B. Jones. \emph{Formal Development Of Correct Algorithms: An Example Based On Earley's Recogniser},
Proving Assertions about Programs, 1972.
\bibitem{marpa}Jeffrey Kegler. \emph{The Marpa parser}, \url{http://jeffreykegler.github.io/Marpa-web-site/}.
\bibitem{rubyslippers}Jeffrey Kegler. \emph{Marpa and the Ruby Slippers}, \url{http://blogs.perl.org/users/jeffrey_kegler/2011/11/marpa-and-the-ruby-slippers.html}.
\bibitem{locallexingtheories}Steven Obua. \emph{Isabelle Theories for Local Lexing (v1.1.0)}, \doi{10.5281/zenodo.322419}.
\bibitem{locallexingprototype}Steven Obua. \emph{Local Lexing Prototype Implementation (v1.0.1)}, \doi{10.5281/zenodo.322417}.
%\bibitem{removeEps} eps
%\bibitem{linearEarley} linear
%\bibitem{practicalEarley} practical
%\bibitem{stixfonts}Stix Fonts. \mbox{\url{http://www.stixfonts.org/}}
\end{thebibliography}

\end{document}